\documentclass[graybox]{svmult}

\usepackage{helvet}         
\usepackage{courier}        
\usepackage{type1cm}        
\usepackage{makeidx}         
\usepackage{graphicx}        
\usepackage{multicol}        
\usepackage[bottom]{footmisc}

\usepackage{rotating}

\usepackage{amssymb,amsfonts}
\usepackage[all,cmtip]{xy}

\def\be{\begin{equation}}
\def\ee{\end{equation}}
\def\bea{\begin{eqnarray}}
\def\eea{\end{eqnarray}}

 \def\lh{{\cal H}_\omega}
  
  \def\qlhz{   C^\infty({\cal H}_{z,\omega}^\ast) }
  \def\sinhc{\,{\rm shc}}
  \def\shc{\,{\rm shc}}
  
  \def\cosh{\,{\rm ch}}
  \def\sinh{\,{\rm sh}}
 
\def\1{\'{\i}}

\def\eee{{\rm e}}

\begin{document}

\title*{A unified approach to Poisson--Hopf deformations of Lie--Hamilton systems based on $\mathfrak{sl}(2)$}
\titlerunning{Poisson--Hopf deformations of  $\mathfrak{sl}(2)$ Lie--Hamilton systems}
\author{\'Angel Ballesteros, Rutwig Campoamor-Stursberg,  {Eduardo Fern\'andez-Saiz}, Francisco J. Herranz and Javier de Lucas}
\authorrunning{\'A. Ballesteros {et al.}}
\institute{\'A. Ballesteros \at Departamento de F\1sica, Universidad de Burgos,   E-09001 Burgos, Spain,\\ \email{angelb@ubu.es}
\and R. Campoamor-Stursberg \at  Instituto de Matem\'atica Interdisciplinar I.M.I-U.C.M, Pza.~Ciencias 3, E-28040 Madrid, Spain, \email{rutwig@ucm.es}
\and E. Fern\'andez-Saiz \at Departamento de Geometr\1a y Topolog\1a,  Universidad Complutense de Madrid, Pza.~Ciencias 3, E-28040 Madrid, Spain, \email{eduardfe@ucm.es}
\and  F.J. Herranz \at Departamento de F\1sica, Universidad de Burgos,   E-09001 Burgos, Spain,\\ \email{fjherranz@ubu.es}
\and  J. de Lucas \at Department of Mathematical Methods in Physics, University of Warsaw, Pasteura 5, 02-093, Warszawa, Poland, \email{ javier.de.lucas@fuw.edu.pl }
}
\maketitle

\abstract{Based on a recently developed procedure to construct Poisson--Hopf deformations of Lie--Hamilton systems~\cite{BCFHL}, a novel unified approach 
to non\-equivalent deformations of Lie--Hamilton systems on the real plane with  a Vessiot--Guldberg Lie algebra isomorphic to $\mathfrak{sl}(2)$ is proposed. This, in particular, allows us 
to define a notion of Poisson--Hopf systems in dependence of a parameterized family of Poisson algebra representations. Such an approach is explicitly illustrated by applying it to the three non-diffeomorphic classes of  $\mathfrak{sl}(2)$ Lie--Hamilton systems. Our results cover deformations of the Ermakov system,    Milne--Pinney,  Kummer--Schwarz and several Riccati equations as well as of the harmonic oscillator (all of them with $t$-dependent coefficients). 
Furthermore $t$-independent constants of motion are given as well.   Our methods can be employed to generate other Lie--Hamilton systems and their deformations for other Vessiot--Guldberg Lie algebras and their deformations.\footnote{Based on the contribution presented at the  ``X International Symposium on {\em Quantum Theory and Symmetries}" (QTS-10), June 19-25, 2017, Varna, Bulgaria}}


\section{Introduction}
\label{sec:1}

Since its original formulation by Lie \cite{LS}, nonautonomous
first-order systems of ordinary differential equations admitting a nonlinear superposition rule, the so-called {\it Lie systems}, have been studied extensively (see 
\cite{C135,CGM00,Diss,Ibrag,Pa57,PW,VES,PWb} and references therein). The Lie theorem \cite{CGM07,LS} states that every system of first-order differential equations is a Lie system if and only if it can be described as a curve in a finite-dimensional Lie algebra of vector fields, a referred to as {\it Vessiot--Guldberg Lie algebra}.

Although being a Lie system is rather an exception than a rule \cite{CGL09,In65,In72}, Lie systems have   been shown to be of great interest within physical and mathematical applications (see \cite{Diss} and references therein). Surprisingly, Lie systems admitting a Vessiot--Guldberg Lie algebra of Hamiltonian vector fields relative to a Poisson structure, the {\it Lie--Hamilton systems}, have found even more applications than standard Lie systems with no associated geometric structure \cite{BBHLS,BCHLS13Ham,BHLS, CLS12Ham}. Lie--Hamilton systems admit an additional finite-dimensional Lie algebra of Hamiltonian functions, a {\it Lie--Hamilton algebra}, that allows for the algebraic determination of superposition rules and constants of motion of the system~\cite{BHLS}.

Apart from the theory of quasi-Lie systems \cite{CGL09} and superposition rules for nonlinear operators   \cite{In65, In72}, most approaches to Lie systems rely strongly in the theory of Lie algebras and Lie groups \cite{IN}. However,  the success of quantum groups \cite{Chari,Majid} and the coalgebra formalism within the analysis of superintegrable systems \cite{coalgebra2,coalgebra3,coalgebra1}, and the fact that quantum algebras appear as deformations of Lie algebras suggested the possibility of extending the notion and techniques of Lie--Hamilton systems beyond the range of application of the Lie theory. An approach in this direction was recently proposed in   
\cite{BCFHL}, where a method to construct quantum deformed Lie--Hamilton systems  (LH systems in short) by means of the coalgebra formalism and quantum algebras was given. 

The underlying idea is to use the theory of quantum groups to deform Lie systems and their associated structures. More exactly, the deformation transforms a LH system with its Vessiot--Guldberg Lie algebra into a Hamiltonian system whose dynamics is determined by a set of generators of a Steffan--Sussmann distribution. Meanwhile, the initial Lie--Hamilton algebra (LH algebra in short)  is mapped into a Poisson--Hopf algebra. The deformed structures allow for the explicit   construction of $t$-independent constants of the motion through quantum algebra techniques for the deformed system. 

This work aims to illustrate how the approach introduced in \cite{BCFHL} to construct deformations of LH systems via Poisson--Hopf structures allows for a further systematization that encompasses the nonequivalent LH systems corresponding to isomorphic LH algebras. 
Specifically, we show that Poisson--Hopf deformations of LH systems based on a LH algebra isomorphic to $\mathfrak{sl}(2)$ can be described 
generically, hence providing the deformed Hamiltonian functions and the corresponding deformed Hamiltonian vector fields, once the corresponding 
counterpart  of the non-deformed system is known. 

Moreover  this work also provides a new method to construct LH systems with a LH algebra
isomorphic to a fixed Lie algebra $\mathfrak{g}$. Our approach relies on using the symplectic foliation in $\mathfrak{g}^*$ induced by the Kirillov--Konstant--Souriou bracket on $\mathfrak{g}$. As a particular case, it is explicitly shown how our procedure explains the existence of three types of LH systems on the plane related to a LH  algebra isomorphic to $\mathfrak{sl}(2)$. This is due to the fact that each one of the three different types corresponds to one of the three types of symplectic leaves in $\mathfrak{sl}^*(2)$. Analogously, one can generate the only type of LH systems on the plane admitting a Vessiot--Guldberg Lie algebra isomorphic to $\mathfrak{so}(3)$.

Our systematization permits us to give directly the Poisson--Hopf deformed system from the classification of LH systems~\cite{BBHLS,BHLS},
further suggesting a notion of Poisson--Hopf Lie systems based on a $z$-parameterized family of Poisson algebra morphisms. Our methods seem to be extensible to study also  LH systems and their deformations on other more general manifolds.

The structure of the contribution goes as follows. Section 2 is devoted to introducing the main aspects of  LH systems and Poisson--Hopf algebras.  The general approach to construct Poisson--Hopf algebra deformations of LH systems~\cite{BCFHL} is summarized in Section 3. For our further purposes,   the (non-standard) Poisson--Hopf algebra deformation of  $\mathfrak{sl}(2)$  is recalled in Section 4.  The novel unifying approach to   deformations of Poisson--Hopf  Lie systems with a LH algebra isomorphic to a fixed Lie algebra $\mathfrak{g}$ are treated in Section 5. Such a procedure is explicitly illustrated in Section 6 by applying it to the three non-diffeomorphic classes of  $\mathfrak{sl}(2)$-LH systems on the plane, so obtaining in a straightforward way their corresponding deformation.
Next,  a new method to construct (non-deformed) LH systems is presented in Section 6.   Finally, our results are summarised and the future work to be accomplished is briefly detailed in the last Section.


\section{Lie--Hamilton systems and Poisson--Hopf algebras }
\label{sec:2}

This section recalls the main notions that will be used in the sequel. 

Let $\left\{x_1,\dots, x_n\right\}$ be global coordinates in $\mathbb{R}^n$ and consider a nonautonomous system of
first-order ordinary differential equations 
\begin{equation}\label{Sys}
\frac{{\rm d} x_k}{{\rm d}t}=f_k(t,x_1,\dots ,x_n), \qquad 1\leq k\leq n,
 \label{system}
\end{equation}
where $f_k:\mathbb{R}^{n+1}\rightarrow \mathbb{R}$ are arbitrary
functions. Geometrically, this system   amounts to a $t$-dependent vector field  ${\bf
X}_t:\mathbb{R}\times \mathbb{R}^n\rightarrow {\rm T}\mathbb{R}^n$
given  by
$$
{\bf X}_t:\mathbb{R}\times\mathbb{R}^n\ni (t,x_1,\dots ,x_n)\mapsto
\sum_{k=1}^{n} f_k(t,x_1,\dots  ,x_n)\frac{\partial}{\partial
x_k} \in {\rm T}\mathbb{R}^n .
$$

We say that (\ref{system}) is a {\em Lie system} if its
general solution, ${\bf x}(t)$, can be expressed in terms of a
finite number $m$ of generic particular solutions
$\left\{  \mathbf{y}_{1}(t),\dots ,\mathbf{y}_{m}(t)\right\}  $ and $n$
constants $\left\{
C_{1},\dots,C_{n}\right\}  $ in the form
$$
{\bf x}(t)=\Psi (  \mathbf{y}_{1}(t),\ldots,\mathbf{y}_{m}(t),C_{1}
,\ldots,C_{n} ), 
$$
for a certain function $\Psi:(\mathbb{R}^n)^{m}\times \mathbb{R}^n\rightarrow \mathbb{R}^n$, a so-called {\it superposition rule} of the 
system (\ref{system}). 

The {Lie--Scheffers Theorem}~\cite{CGM00,CGM07,LS,VES} states that $\mathbf{X}_t$ is a Lie system if and only if there exist $t$-dependent functions $b_1(t),\ldots,b_r(t)$ and vector fields ${\bf X}_1,\ldots,{\bf X}_r$ on $\mathbb{R}^n$
spanning an $r$-dimensional real Lie algebra $V$ 
such that 
$$
{\bf X}_t(x,y)=\sum_{i=1}^r b_i(t){\mathbf X}_i.\label{aabb} 
$$
Then, $V$ is called a {\em Vessiot--Guldberg Lie algebra} of ${\bf X}_t$. 

A Lie system is said to be a {\em Lie--Hamilton system} \cite{CLS12Ham} whenever it admits a Vessiot--Guldberg Lie algebra $V$ of Hamiltonian vector fields with respect to a Poisson structure. In our work, we will focus on LH systems on the plane admitting a Vessiot--Guldberg Lie algebra of Hamiltonian vector fields relative to a symplectic structure. It can be proved that all LH systems can be studied around a generic point in this way \cite{BBHLS}.

Hence, the LH systems to be studied hereafter admit a symplectic structure $\omega$ on $\mathbb{R}^2$ that is invariant under Lie derivatives with respect to the elements of $V$, namely
$$
L_{\mathbf{X}_i}\omega =0, \qquad 1 \leq i\leq r.\label{symp}
$$ 
Due to the non-degeneracy of $\omega$, each function $h$ determines uniquely a vector field ${\bf X}_h$, the {\it Hamiltonian vector field} of $h$, such that $\iota_{{\bf X}_h}\omega={\rm d}h$, enabling us to define a Poisson bracket 
$$
\{\cdot,\cdot\}_\omega\ :\ C^\infty (\mathbb{R}^{2} )\times C^\infty (\mathbb{R}^{2} )\rightarrow C^\infty (\mathbb{R}^{2} )
$$
 through the prescription 
 \begin{equation}\label{LB}
\{f,g\}_\omega\mapsto {\bf X}_{g} f.
 \end{equation}
In particular, this implies that $(C^\infty(\mathbb{R}^{2}),\{\cdot,\cdot\}_\omega)$ is a Lie algebra. Similarly, the space ${\rm Ham}(\omega)$ of  Hamiltonian vector fields on $\mathbb{R}^{2}$ relative to $\omega$ is a Lie algebra with respect to the commutator of vector fields. These two Lie algebras are known to be related through the exact sequence (see \cite{Va94} for details): 
$$\label{seq}
0\hookrightarrow \mathbb{R}\hookrightarrow (C^\infty(\mathbb{R}^{2}),\{\cdot,\cdot\}_\omega)\stackrel{\varphi}{\longrightarrow} ({\rm Ham}(\omega),[\cdot,\cdot])\stackrel{\pi}{\longrightarrow} 0,
$$
where $\varphi$ maps every function $h\in C^\infty(\mathbb{R}^2)$ into $-{\bf X}_h$. 

Going back to the theory of LH systems, recall that every LH system admits a Vessiot--Guldberg Lie algebra $V$ of Hamiltonian vector fields relative to an $\omega$. In view of (\ref{LB}), there always exists a finite-dimensional Lie subalgebra $\lh$ of $(C^\infty(\mathbb{R}^{2}),\{\cdot,\cdot\}_\omega)$ containing the Hamiltonian functions of $V$: a so-called {\em  Lie--Hamilton  algebra}  of  the LH system ${\bf X}_t$. 

Let $\mathfrak{g}$ be a Lie algebra isomorphic to $\lh$. This induces the {\it universal enveloping algebra} $U(\mathfrak{g})$ and the {\it symmetric algebra} $S(\mathfrak{g})$ (see \cite{Va84} for details). The second one is the associative commutative algebra of polynomials in the elements of $\mathfrak{g}$, whereas $U(\mathfrak{g})$ is defined to be the tensor algebra of $\mathfrak{g}$ modulo the two-sided ideal generated by the elements $\{v\otimes w-w\otimes v-[v,w]:v,w\in \mathfrak{g}\}$. 

Relevantly, $S(\mathfrak{g})$ and $U(\mathfrak{g})$ are isomorphic as linear spaces \cite{Va84}. They also share a special property: they are Hopf algebras. The Lie bracket of $\mathfrak{g}$ can be extended to $S(\mathfrak{g})$ turning this space into a Poisson algebra. Since the elements of $\mathfrak{g}$ can be considered as linear functions on $\mathfrak{g}^*$, then the elements of $S(\mathfrak{g})$ can be considered as elements of $C^\infty(\mathfrak{g}^*)$, which allows us to ensure that the space $C^\infty (\mathfrak{g}^*  )$ can be endowed with a {\it Poisson--Hopf algebra} structure. 
 
Let us finally recall in this introduction the main properties of Hopf algebras. We recall that an associative algebra $A$  with a {\it product} $m$ and a {\it unit} $\eta$  is said to be a {\em Hopf algebra} over $\mathbb R$
\cite{Abe,Chari, Majid} if there exist two homomorphisms called {\em coproduct}  $(\Delta :
A\longrightarrow A\otimes A )$ and {\em counit} $(\epsilon : A\longrightarrow
\mathbb R)$ satisfying  
$$
({\rm Id}\otimes\Delta)\Delta=(\Delta\otimes {\rm Id})\Delta,  \qquad
({\rm Id}\otimes\epsilon)\Delta=(\epsilon\otimes {\rm Id})\Delta={\rm Id}, 
$$
along with an  antihomomorphism, the {\em antipode} $\gamma :
A\longrightarrow A$,  such that
the following   diagram is commutative:
\medskip

   \centerline{ {\xymatrix@C=0.4em{&A\otimes A\ar[rr]^{{\rm Id}\,\otimes\,  \gamma}&&A\otimes A\ar[rd]^m&&\\
   A\ar[rd]^\Delta\ar[rr]^{\epsilon}\ar[ur]^\Delta&&\mathbb{R}\ar[rr]^{\eta}&&A\\
   &A\otimes A\ar[rr]^{\gamma\,\otimes\, {\rm Id}}&&A\otimes A\ar[ru]^m&}	}}


\section{Poisson--Hopf deformations of Lie--Hamilton   systems}
\label{sec3}

The coalgebra method employed in \cite{BCHLS13Ham} to obtain superposition rules and constants of motion for LH systems on a manifold $M$ relies almost uniquely in the Poisson--Hopf algebra structure related to $C^\infty(\mathfrak{g}^*)$ and a Poisson map 
$$
D:C^\infty(\mathfrak{g}^*)\rightarrow C^\infty(M),$$ 
where we recall that $\mathfrak{g}$ is a Lie algebra isomorphic to a LH algebra, $\mathcal{H}_\omega$, of the LH system.  

Relevantly, quantum deformations allow us to repeat this scheme by substituting  the Poisson algebra $C^\infty(\mathfrak{g}^*)$ with a quantum deformation $C^\infty(\mathfrak{g}_z^*)$, where $z\in \mathbb{R}$, and obtaining an adequate Poisson map 
$$
D_z:C^\infty(\mathfrak{g}_z^*)\rightarrow C^\infty(M).
$$

The above procedure enables us to deform the LH system into a $z$-parametric family of Hamiltonian systems whose dynamic is determined by a Steffan--Sussmann distribution and a family of Poisson algebras. If  $z$ tends to zero, then the properties of the (classical) LH system are recovered by a limiting process, hence enabling to construct new deformations exhibiting physically relevant 
properties. 

In essence, the method for a LH system on an $n$-dimensional  manifold $M$ consists essentially of the following four steps (see  \cite{BCFHL} for details):

\begin{enumerate}

\item Consider a LH system ${\bf X}_t:=\sum_{i=1}^r b_i(t){\bf X}_i$ on  $M$ with respect to a symplectic form $\omega$ and possessing a LH algebra $\mathcal{H}_{\omega}$   spanned by the functions $\left\{h_1,\ldots, h_r\right\}\subset C^\infty(M)$ and structure constants $C_{ij}^k$, i.e.
$$
\{h_i,h_j\}_{\omega}= \sum_{{k=1}}^{r} C_{ij}^{k}h_{k},\qquad 1\leq i,j\leq r.
\label{za}
$$

\item Consider a Poisson--Hopf algebra deformation $\qlhz$ with (quantum) deformation parameter $z\in\mathbb R$ (respectively $q:={\rm e}^z$) as the space of smooth functions $F(h_{z,1},\ldots,h_{z,r})$ for a family of functions $h_{z,1},\ldots,h_{z,r}$  on $M$  such that
\be
\{h_{z,i},h_{z,j}\}_{\omega}= F_{z,ij}(h_{z,1},\dots,h_{z,r}),
\label{zab}
\ee
where the $F_{z,ij}$ are  smooth functions depending also on $z$ satisfying the boundary conditions  
\bea
&& \lim_{z\to 0} h_{z,i}=h_i ,\qquad  \lim_{z\to 0} {\rm grad}\, h_{z,i}={\rm grad} \, h_i,  \nonumber\\
&&   \lim_{z\to 0}   F_{z,ij}(h_{z,1},\dots,h_{z,r }) = \sum_{{k=1}}^{r} C_{ij}^k h_k   .
\label{zac}
\eea

\item Define the deformed vector fields ${\bf X}_{z,i}$ on $M$ according to the rule
\be
\iota_{{\bf X}_{z,i}}\omega={\rm d}h_{z,i},
\label{iot}
\ee
so that
\begin{equation}
\lim_{z\to 0} {\bf X}_{z,i}= {\bf X}_i.
\label{zae}
\end{equation}
\item Define the Poisson--Hopf deformation of the LH system ${\bf X}_t$ as 
$$
{\bf X}_{z,t}:=\sum_{i=1}^r b_i(t){\bf X}_{z,i}.
$$
\end{enumerate}

We stress that  the deformed vector fields $\{ {\bf X}_{z,1},\dots,  {\bf X}_{z,r} \}$ do not generally close on a finite-dimensional Lie algebra. Instead, they span a Stefan--Sussman distribution (see \cite{WA,Pa57,Va94}). Their corresponding commutation relations can be written in terms of the functions $F_{z,ij}$ as~\cite{BCFHL}
\be
[{\bf X}_{z,i},{\bf X}_{z,j}]= -\sum_{k=1}^r\frac{\partial F_{z,ij}}{\partial h_{z,k}}\, {\bf X}_{z,k}.
\label{coff}
\ee

Next, 
to determine the $t$-independent constants of the motion and the superposition rules of a LH system with a LH algebra $\mathcal{H}_\omega$, the coalgebra formalism developed in \cite{BCHLS13Ham} is applied. Let us illustrate this point. Consider the symmetric algebra $S\left(\mathfrak{g}\right)$ of $\mathfrak{g}\simeq\lh$, that can be endowed with a Poisson algebra structure by means of the Lie algebra structure of $\mathfrak{g}$. The Hopf algebra structure with a (non-deformed  trivial) coproduct map $\Delta$ is given  by
$$
 {\Delta} :S\left(\mathfrak{g}\right)\rightarrow
S\left(\mathfrak{g}\right) \otimes S\left(\mathfrak{g}\right)    ,\qquad      {\Delta}(v):=v\otimes 1+1\otimes v,  \qquad    \forall v\in \mathfrak{g}.
\label{baa}
$$
This is easily seen to be a Poisson algebra homomorphism with respect to the Poisson structure on $S(\mathfrak{g})$ and the natural Poisson structure in $S(\mathfrak{g})\otimes S(\mathfrak{g})$ induced by $S(\mathfrak{g})$. Due to density of the functions $S(\mathfrak{g})$ in $C^\infty(\mathfrak{g}^*)$,  the coproduct $\Delta$ can be extended in a unique way to  
$$
 {\Delta} :C^\infty\left(\mathfrak{g}^*\right)\rightarrow
C^\infty\left(\mathfrak{g}^*\right) \otimes C^\infty\left(\mathfrak{g}^*\right).
\label{baa+}
$$
The extension by continuity of the Poisson--Hopf structure in $S(\mathfrak{g})$ to $C^\infty(\mathfrak{g}^*)$ endows the latter with a Poisson--Hopf algebra structure~\cite{BCHLS13Ham}.

Let now  $C=C(v_1,\dots,v_r)$ be a Casimir function of the Poisson algebra $C^\infty\left(\mathfrak{g}^*\right)$, where $v_1,\ldots,v_r$ is a basis for $\mathfrak{g}$. We can define a Lie algebra morphism $\phi:\mathfrak{g}\rightarrow C^\infty(M)$ such that  $h_i:=\phi(v_i)$. The Poisson algebra morphisms 
$$
D: C^\infty\left( \mathfrak{g}^* \right) \rightarrow C^\infty(M),\quad\  D^{(2)} :   C^\infty\left(  \mathfrak{g}^* \right)\otimes C^\infty\left( \mathfrak{g}^* \right)\rightarrow C^\infty(M)\otimes C^\infty(M),
\label{morphisms}
$$
defined by
$$
D( v_i):= h_i({\bf x}_1), \qquad
 D^{(2)} \left( {\Delta}(v_i) \right):= h_i({\bf x}_1)+h_i({\bf x}_2)   ,\qquad 1\leq i\leq  r,
\label{bb}
$$
where
  ${\bf x}_s=\left\{x_{s,1},\dots  ,x_{s,n}\right\}$ $(s=1,2)$ are global coordinates in $M$,  lead to $t$-independent constants of the motion $F^{(1)}:= F$ and $F^{(2)}$  of ${\bf X}_t$  having the form
\be
  F:= D(C),\qquad F^{(2)}:=  D^{(2)} \left( {\Delta}(C) \right).
\label{bc}
\ee

The  very same argument holds to any deformed Poisson--Hopf algebra $C^\infty(\mathfrak{g}_z^*)$ with deformed coproduct $\Delta_z$  and Casimir invariant
$C_z=C_z(v_{z,1},\dots,v_{z,r})$, where $\{ v_{z,1},\dots,v_{z,r} \}$  satisfy the same formal  commutation relations of the $h_{z,i}$ in ({\ref{zab}), and such that
$$
\lim_{z\to 0} \Delta_{z}=\Delta ,\qquad \lim_{z\to 0} v_{z,i}= v_i, \qquad \lim_{z\to 0} C_{z}= C.
$$
Therefore, the deformed Casimir $C_z$ will provide the $t$-independent constants of motion for the deformed LH system ${\bf X}_{z,t}$ through the coproduct $\Delta_z$.


\section{The non-standard Poisson--Hopf algebra deformation of  $\mathfrak{sl}(2)$}
Amongst the LH systems in the plane (see \cite{BBHLS,BCHLS13Ham,BHLS} for details and applications), those with a Vessiot--Guldberg Lie algebra isomorphic to $\mathfrak{sl}(2)$ are of both mathematical and physical  interest; they cover complex Riccati, Milne--Pinney and Kummer--Schwarz equations as well as the harmonic oscillator, all of them with $t$-dependent coefficients.
Furthermore, $\mathfrak{sl}(2)$-LH systems   are related to three non-diffeomorphic Vessiot--Guldberg Lie algebras on the plane \cite{BBHLS,BHLS}. This gives rise to different nonequivalent Poisson--Hopf deformations.

Let us consider  $\mathfrak{sl}(2)$ with the standard basis $\{J_3, J_+,J_-\}$ satisfying the commutation relations
$$
[ J_3,J_\pm  ]  = \pm  2J_\pm     ,\qquad  [J_+ , J_- ] = J_3.
$$
In this basis, the Casimir operator reads 
\be
\mathcal{C}=\frac{1}{2}  J_3^2+(J_+ J_- + J_- J_+).
\label{cas1}
\ee
 Considering the non-standard (triangular or Jordanian) quantum deformation $U_z(\mathfrak{sl}(2))$ of $\mathfrak{sl}(2)$~\cite{Ohn}  (see also  \cite{non,beyond} and references therein), we are led to the following deformed coproduct 
\bea
&& \Delta_z(J_{+})=  J_+ \otimes 1+
1\otimes J_+ , \nonumber\\
&&\Delta_z(J_l)=J_l \otimes \eee^{2 z J_+} + \eee^{-2 z J_+} \otimes
J_l ,\qquad l\in\left\{-,3\right\}
\nonumber
\eea
and the commutation rules 
\bea
&& \left[J_3,J_+\right]_z= 2\,{{\rm shc}(2z J_+)}  J_{+}, \qquad\left[J_+,J_-\right]_z=J_3, \nonumber\\
&& \left[J_3,J_-\right]_z=-J_{-} \,{\rm ch}(2z J_+)-{\rm ch}(2z J_+)J_{-}.
\nonumber
\eea
Here ${\rm shc}$ denotes the cardinal hyperbolic sinus function defined by 
$$
\shc( \xi):=  \left\{ 
\begin{array}{ll}
\frac{\sinh (\xi)}{\xi}, & \ \  \mbox{for}\  \xi\ne 0, \\
1,& \ \  \mbox{for}\  \xi=0.
\end{array}
 \right. 
$$

It is known that every quantum algebra $U_z(\mathfrak{g})$ related to a semi-simple Lie algebra $\mathfrak{g}$ admits an isomorphism of algebras $U_z(\mathfrak{g})\rightarrow U(\mathfrak{g})$ (see \cite[Theorem 6.1.8]{Chari}). This allows us to obtain a Casimir operator of $U_z(\mathfrak{sl}(2))$ out of one, e.g. $\mathcal{C}$ (\ref{cas1}), of $U(\mathfrak{sl}(2))$ in the form (see \cite{Chari} for details)
$$
\mathcal{C}_z=\frac{1}{2}J_3^{2}+  {{\rm shc}( 2z J_{+})}{J_+}J_{-} + J_{-}{J_+} \, {{\rm shc}( 2z J_{+})} +\frac{1}{2}\,{\rm ch}^{2}(2zJ_{+}),
\label{gx}
$$
which, as expected, coincides with  the expression  formerly given in~\cite{beyond}.

There exists a new $z$-parametrized family of deformed Poisson--Hopf structures in  $C^\infty(\mathfrak{sl}_z^*(2))$ denoted by $(C^\infty(\mathfrak{sl}_z^*(2)),\{\cdot,\cdot\}_z)$ and given by the relations 
\begin{equation}
\{v_1,v_2\}_z=-\shc (2z v_1)v_1,\quad 
 \{v_1,v_3\}_z=-2 v_2,\quad
\{v_2,v_3\}_z= -  \cosh(2 z v_1)v_3 ,
\label{gb}
\end{equation}
along with the coproduct 
\be
\Delta_z(v_1)=  v_1 \otimes 1+
1\otimes v_1 ,\quad\ 
\Delta_z(v_k)=v_k \otimes \eee^{2 z v_1} + \eee^{-2 z v_1} \otimes
v_k ,\quad\  k=2,3.
\label{gax}
\ee
The Poisson algebra $C^\infty(\mathfrak{sl}_z(2))$   admits  a Casimir function
\be
 {C}_z=\shc( 2z v_1)\, v_1v_3-v_2^2   ,
\label{gc}
\ee
where  
\be
v_1=   J_+,\qquad v_2 =  \frac{1}{2} J_3,\qquad v_3= -  J_- .\label{box}
\ee

In the limit $z=0$, the Poisson--Hopf structure in $C^\infty(\mathfrak{sl}_z^*(2))$ recovers the standard Poisson--Hopf algebra structure in $C^\infty(\mathfrak{sl}^*(2))$ with    non-deformed coproduct  and  Poisson bracket 
\bea
&& {\Delta}(v_i)=v_i\otimes 1+1\otimes v_i ,\qquad i=1,2,3 ,\nonumber\\
&&\{ v_1,v_2\}=- v_1,\qquad \{ v_1,v_3\} =- 2v_2,\qquad  \{ v_2,v_3\} =-v_3,
\label{brack2}
\eea
as well as the Casimir function 
\be {C}=v_1 v_3 - v_2^2 .
\label{cas2}
\ee

We shall make use of the above Poisson--Hopf algebra  $C^\infty(\mathfrak{sl}_z^*(2))$ in the next section in order to construct
the corresponding deformed LH systems from a unified approach.


\section{Poisson--Hopf deformations of   $\mathfrak{sl}(2)$ Lie--Hamilton systems}
\label{PHL}
 
This section concerns the analysis of Poisson--Hopf deformations of LH systems on a manifold $M$  with a Vessiot--Guldberg algebra isomorphic to $\mathfrak{sl}(2)$. Our geometric analysis will allow us to
introduce the notion of a Poisson--Hopf Lie system that,  roughly speaking,   is a family of nonautonomous Hamiltonian systems of first-order differential equations constructed as a deformation of a LH system by means of the representation of the deformation of a Poisson--Hopf algebra in a Poisson manifold. 
 
Let us  endow a manifold $M$ with a symplectic structure $\omega$ and consider a Hamiltonian Lie group action $\Phi:SL(2,\mathbb{R})\times M\rightarrow M$. A basis of fundamental vector fields of  $\Phi$, let us say $\left\{{\bf X}_1,{\bf X}_2,{\bf X}_3\right\}$, enable us to define a Lie system 
$$
{\bf X}_t=\sum_{i=1}^3b_i(t){\bf X}_i,
$$
for arbitrary $t$-dependent functions $b_1(t),b_2(t),b_3(t)$, and    $\{{\bf X}_1,{\bf X}_2,{\bf X}_3\}$ spanning a Lie algebra isomorphic to $\mathfrak{sl}(2)$. As is well known, there are only three non-diffeomorphic classes of Lie algebras of Hamiltonian vector fields isomorphic to $\mathfrak{sl}(2)$ on the plane \cite{BBHLS,GKO92}. 
Since ${\bf X}_1,{\bf X}_2,{\bf X}_3$ admit Hamiltonian functions $h_1,h_2,h_3$, the $t$-dependent vector field ${\bf X}$ admits a $t$-dependent Hamiltonian function 
$$
h=\sum_{i=1}^3b_i(t)h_i.
$$
Due to the cohomological properties of $\mathfrak{sl}(2)$ (see e.g. \cite{Va94}), the Hamiltonian functions $h_1,h_2,h_3$ can always be chosen so that the space $\langle h_1,h_2,h_3\rangle$ spans a Lie algebra isomorphic to $\mathfrak{sl}(2)$ with respect to $\{\cdot,\cdot\}_\omega$.

Let $\{v_1,v_2,v_3\}$ be the basis for $\mathfrak{sl}(2)$ given in (\ref{box}) and let $M$ be a manifold where the functions $h_1,h_2,h_3$ are smooth. Further, the Poisson--Hopf algebra structure of  $C^\infty(\mathfrak{sl}^*(2))$ is given by (\ref{brack2}). In these conditions, there exists a Poisson algebra morphism $D:C^\infty(\mathfrak{sl}^*(2))\rightarrow C^\infty(M)$ satisfying
$$
D(f(v_1,v_2,v_3))=f(h_1,h_2,h_3),\quad \forall f\in C^\infty(\mathfrak{sl}^*(2)).
$$
Recall that the deformation $C^\infty(\mathfrak{sl}_z^*(2))$ of $C^\infty(\mathfrak{sl}^*(2))$ is a Poisson--Hopf algebra with the new Poisson structure induced by the relations (\ref{gb}). Let us define the submanifold $\mathcal{O}=:\{\theta\in \mathfrak{sl}^*(2): v_1(\theta)\neq0\}$ of $\mathfrak{sl}^*(2)$.  Then, the Poisson structure on $\mathfrak{sl}^*(2)$ can be restricted to the space $C^\infty(\mathcal{O})$. In turn, this enables us to expand the Poisson--Hopf algebra structure in $C^\infty(\mathfrak{sl}^*(2))$ to $C^\infty(\mathcal{O})$. Within the latter space, the elements 
\bea
&& v_{z,1}:=v_1,\qquad v_{z,2}:={\sinhc}(2z v_{1})v_{2},\nonumber\\
&&v_{z,3}:={\sinhc}(2zv_1)\frac{v^2_2}{v_1}+\frac{c}{4{\sinhc}(2zv_1)v_1} ,
\label{nf}
\eea
are easily verified to satisfy the same commutation relations with respect to $\{\cdot,\cdot\}$ as the elements $v_1,v_2,v_3$ in $C^\infty(\mathfrak{sl}_z^*(2))$ with respect to $\{\cdot,\cdot\}_z$ (\ref{gb}), i.e.
\bea
&&
\{v_{z,1},v_{z,2}\}=-{\sinhc}(2zv_{z,1})v_{z,1}, \qquad \{v_{z,1},v_{z,3}\}=-2v_{z,2}, \nonumber\\
&&\{v_{z,2},v_{z,3}\}=-\,{\rm ch}(2zv_{z,1}) v_{z,3}   . 
\label{ggbb}
\eea
In particular, from (\ref{nf}) with $z=0$ we find that
\bea
&&
\{v_{0,1},v_{0,2}\}=-v_{0,1},\qquad
\{v_{0,1},v_{0,3}\}=\frac{2\{v_{0,1},v_{0,2}\} v_{0,2}}{v_{0,1}}=-2v_{0,2},  \nonumber\\
&&\{v_{0,2},v_{0,3}\}=-\frac{v_{0,2}^2}{v_{0,1}^2}\{v_{0,2},v_{0,1}\}-\frac{c}{4v_{0,1}^2}\{v_{0,2},v_{0,1}\}=-\frac{v_{0,2}^2}{v_{0,1}}-\frac{c}{4v_{0,1}}=-v_{0,3}. 
\nonumber
\eea
The functions $v_{z,1},v_{z,2},v_{z,3}$ are not functionally independent, as they satisfy the constraint 
\begin{equation}
{\sinhc}(2zv_{1,z})v_{1,z}v_{3,z}-v^2_{2,z}=c/4.\label{kas}
\end{equation}

The existence of the functions $v_{z,1},v_{z,2},v_{z,3}$ and the relation (\ref{kas}) with the Casimir of the deformed Poisson--Hopf algebra is by no means casual. Let us explain why $v_{z,1},v_{z,2},v_{z,3}$ 
exist and how to obtain them easily. 

Around a generic point $p\in \mathfrak{sl}^*(2)$, there always exists  an open $U_p$ containing $p$ where both Poisson 
structures give a symplectic foliation by surfaces.  Examples of symplectic leaves for $\{\cdot,\cdot\}$ and $\{\cdot,\cdot\}_z$ are displayed in Fig. \ref{Figure1}.

\begin{figure}[t]
\begin{center}
\includegraphics[scale=0.5]{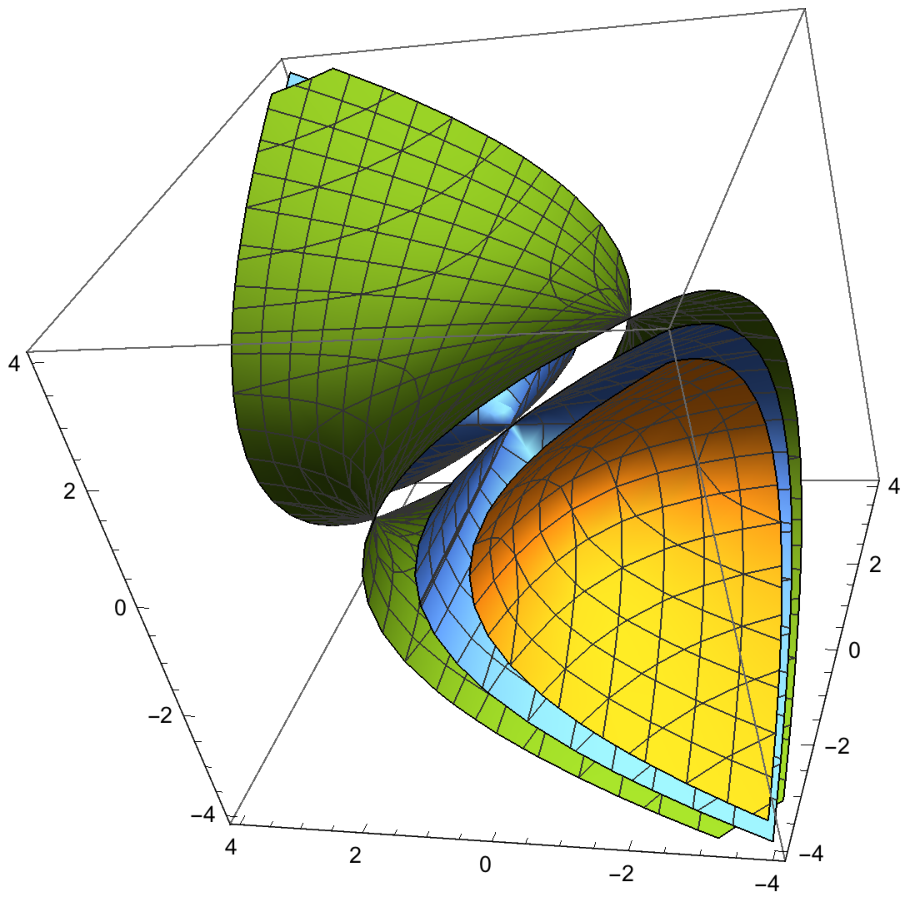}$\qquad\qquad$
\includegraphics[scale=0.5]{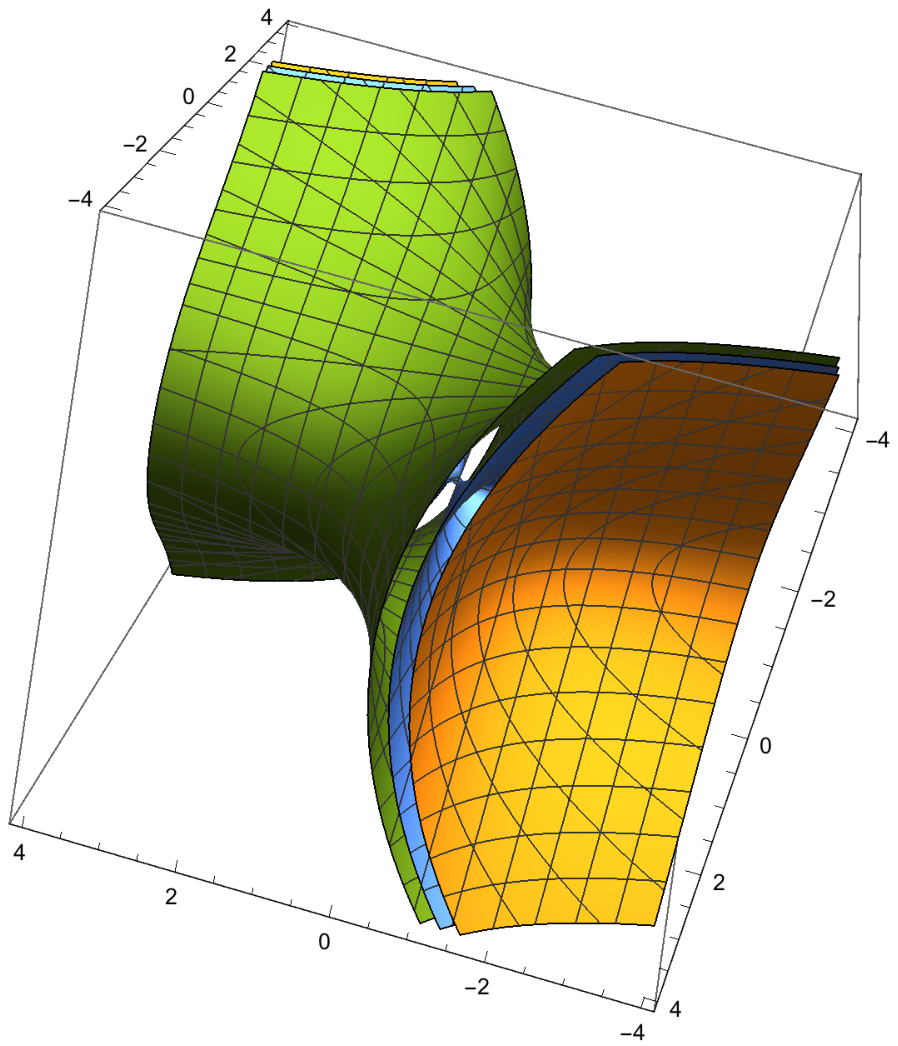}
\end{center}
\caption{Representatives of the submanifolds in $\mathfrak{sl}^*(2)$ given by the surfaces with constant value of the Casimir for the Poisson structure in $\mathfrak{sl}^*(2)$ (left) and its deformation (right). Such submanifolds are symplectic submanifolds where the Poisson bivectors $\Lambda$ and $\Lambda_z$ admit a canonical form.}
\label{Figure1}        
\end{figure}

The splitting theorem on 
Poisson manifolds  \cite{Va94} ensure that if $U_p$ is small enough, then there exist two different coordinate systems $\{x,y,C\}$ and $\{x_z,y_z,C_z\}$ where the Poisson bivectors related to $\{\cdot,\cdot\}$ and $\{\cdot,\cdot\}_z$ read $\Lambda=\partial_{x}\wedge \partial_y$ and $\Lambda_z=\partial_{x_z}\wedge 
\partial_{y_z}$.   Hence, $C_z$ and $C$ are Casimir functions for $\Lambda_z$ and $\Lambda$, respectively. Moreover, $x_z=x_z(x,y,C),y_z=y_z(x,y,C), C_z=C_z(x,y,C)$. 
It follows from this that 
$$
\Phi:f(x_z,y_z,C_z)\in C^\infty_z(U_p)\mapsto f(x,y,C)\in C^\infty(U_p)
$$
 is a Poisson algebra morphism. 

If $\left\{v_1,v_2,v_3\right\}$ are the standard coordinates on $\mathfrak{sl}^*
(2)$ and the relations (\ref{gb}) are satisfied, then $v_i=\xi_i(x_z,y_z,C_z)$ holds for certain functions $\xi_1,\xi_2,\xi_3:\mathbb{R}^3\rightarrow \mathbb{R}$. Hence, the $\hat v_{z,i}=\xi_i(x,y,C)$ close the same commutation 
relations relative to $\{\cdot,\cdot\}$ as the $v_i$ do with respect to $\{\cdot,\cdot\}_z$. As $C$ is a Casimir invariant, the functions $v_{z,i}:=\xi_i(x,y,c)$, with a constant value of $c$, still close the same commutation relations among 
themselves as the $v_i$. Moreover, the functions $v_{z,i}$ become functionally dependent. Indeed,
$$
C_z=C_z(v_1,v_2,v_3)=C_z(\xi_1(x_z,y_z,C_z),\xi_2(x_z,y_z,C_z),\xi_3(x_z,y_z,C_z)).
$$
Hence,
$c=C_z(\xi_1(x_z,y_z,c),\xi_2(x_z,y_z,c),\xi_3(x_z,y_z,c))$ and we conclude that $c=C_z(v_{z,1},v_{z,2},v_{z,3}).$ 

The previous argument allows us to recover the functions (\ref{nf}) in an algorithmic way. Actually, the 
functions $x_z,y_z,C_z$ and $x,y,C$ can be easily chosen to be 
$$
x_z:=v_1,\qquad y_z:=-\frac{v_2}{{\sinhc}(2zv_1)v_1},\qquad C_z:={\sinhc}(2zv_1)v_1v_3-v_2^2,
$$
 as well as 
 $$
 x=v_1,\qquad y=-v_2/v_1,\qquad C=v_1v_3-v_2^2.
 $$
Therefore, 
\bea
&&
\xi_1(x_z,y_z,C_z)=x_z,\qquad \xi_2(x_z,y_z,C_z)=-y_z{\sinhc}(2zx_z)x_z, \nonumber\\
&& \xi_3(x_z,y_z,C_z)=\frac{ C_z+x_z^2y_z^2{\sinhc}^2(2zx_z)}{ {\sinhc}(2zx_z)x_z}.
\nonumber
\eea
Assuming that $C_z=c/4$,  
replacing $x_z,y_z$ by $x=v_1,y=-v_2/v_1$, respectively, and taking into account that $v_{z,i}:=\xi_i(x,y,c)$, 
 one retrieves (\ref{nf}).  

It is worth mentioning that due to the simple form of the Poisson bivectors in splitting  form for three-dimensional Lie algebras, 
this method can be easily applied to such a type of Lie algebras.

Next, the above relations enable us to construct a Poisson algebra morphism 
$$
D_z:f(v_1,v_2,v_3)\in C^\infty(\mathfrak{sl}_z^*(2))\mapsto D(f(v_{1,z},v_{2,z},v_{3,z}))\in C^\infty(M)
$$
for every value of $z$ allowing us to pass the structure of the Poisson--Hopf algebra $C^\infty(\mathfrak{sl}_z^*(2))$ to $C^\infty(M)$. As a consequence, $D_z( {C}_z)$ satisfies the relations 
$$
\{D_z( {C}_z),h_{z,i}\}_\omega=0,\qquad i=1,2,3.
$$
Using the symplectic structure on $M$ and the functions $h_{z,i}$ written in terms of $\left\{h_1,h_2,h_3\right\}$, one can easily obtain the deformed vector fields ${\bf X}_{z,i}$ in terms of the vector fields ${\bf X}_i$. Finally, as ${\bf X}_{z,t}=\sum_{i=1}^3b_i(t){\bf X}_{z,i}$ holds, it is straightforward to verify that the brackets  
$$
{\bf X}_{z,i}D(C_z)=\{D(C_z),h_{z,i}\}=0,
$$
imply that the function $D(C_z)$ is a $t$-independent constant of the motion for each of the deformed LH system ${\bf X}_{z,t}$.

Consequently, deformations of LH-systems based on $\mathfrak{sl}(2)$ can be treated simultaneously, starting from their classical 
LH counterpart. The final result  is summarized in the following statement.

\begin{theorem}\label{MT} If $\phi:\mathfrak{sl}(2)\rightarrow C^\infty(M)$ is a morphism of Lie algebras with respect to the Lie bracket in $\mathfrak{sl}(2)$ and a Poisson bracket in $C^\infty(M)$, then for each $z\in \mathbb{R}$ there exists a Poisson algebra morphism $D_z:C^\infty(\mathfrak{sl}_z^*(2))\rightarrow C^\infty(M)$ such that for a basis $\{v_1,v_2,v_3\}$ satisfying the commutation relations (\ref{brack2}) is given by
\bea
 \!  \! && D_z(f(v_1,v_2,v_3))\nonumber\\
 \!  \! &&=f\! \left( \!  \phi(v_1),{\sinhc}(2z \phi(v_1))\phi(v_{2}), {\sinhc}(2z\phi(v_1))\frac{\phi^2(v_2)}{\phi(v_1)}+\frac{c}{4{\sinhc}(2z\phi(v_1))\phi(v_1)}\right). 
\nonumber
\eea
 Provided that $h_i:=\phi(v_i)$,  the deformed Hamiltonian functions  $h_{z,i}:=D_z(v_i)$ adopt the form
\bea
&&h_{z,1}=    h_1 ,\qquad  h_{z,2}=  {\sinhc}(2z  h_1 ) h_{2}    ,  \nonumber\\
&&h_{z,3}=    {\sinhc}(2z h_1)\frac{h_2^2 }{ h_1}+\frac{c}{4{\sinhc}(2z h_1) h_1}  ,
\nonumber
\eea
which satisfy the    commutation relations (\ref{ggbb}).

The Hamiltonian vector fields ${\bf X}_{z,i}$ associated with  $h_{z,i}$ through (\ref{iot}) turn out to be
\bea
&&
{\bf X}_{z,1}={\bf X}_{1},\qquad {\bf X}_{z,2}=\frac{h_2}{h_1}\bigl(\!{\cosh(2zh_1)}-{{\sinhc}(2zh_1)}\bigr){\bf X}_1+{\sinhc}(2zh_1){\bf X}_2, \nonumber\\
&&
{\bf X}_{z,3}= \left[\frac{h_2^2}{h_1^2}\bigl( {\rm ch}(2z h_1)-2 \,{\rm shc}(2z h_1)\bigr)- \frac{c\; {\rm ch}(2z h_1)}{4h_1^2 {\rm shc}^2(2z h_1)}\right] {\bf X}_1\nonumber\\
&&\qquad \qquad\qquad +2\,\frac{h_2}{h_1}\,{\sinhc}(2zh_1){\bf X}_2 ,\nonumber
\eea
and satisfy the following commutation relations coming from (\ref{coff})
\bea
&&
[{\bf X}_{z,1},{\bf X}_{z,2}]={\rm ch}(2zh_{z,1}){\bf X}_{z,1},\qquad [{\bf X}_{z,1},{\bf X}_{z,3}]=2\,{\bf X}_{z,2},\nonumber\\
&&
[{\bf X}_{z,2},{\bf X}_{z,3}]={\rm ch}(2zh_{z,1}){\bf X}_{z,3}+4 z^2{\rm shc}^2(2z h_{z,1}) h_{z,1}h_{z,3}\,{\bf X}_{z,1} .
\nonumber
\eea
\end{theorem}

  As a consequence, the deformed Poisson--Hopf system can be generically described in terms of the Vessiot--Guldberg Lie algebra corresponding to the non-deformed LH system as follows: 
\bea
&&{\bf X}_{z,t}=\sum_{i=1}^3 b_i(t) {\bf X}_{z,i}=
 \left[ b_1(t)+ b_2(t)\,\frac{h_2}{h_1}\bigl( \!{\cosh(2zh_1)}-{{\sinhc}(2zh_1)}\bigr)\right] {\bf X}_1\nonumber\\
&&\qquad\qquad\qquad  +b_3(t)\left[\frac{h_2^2}{h_1^2}\bigl( {\rm ch}(2z h_1)-2\, {\rm shc}(2z h_1)\bigr)- \frac{c\; {\rm ch}(2z h_1)}{4 h_1^2 {\rm shc}^2(2z h_1)}\right]{\bf X}_1 \nonumber\\
&&\qquad\qquad\qquad   +\,{\rm shc}(2zh_1)\left(b_2(t)+ 2b_3(t)\frac{ h_2}{h_1}\right) {\bf X}_2. \label{DYF}\nonumber
\eea

This unified approach to nonequivalent deformations of LH systems possessing a common underlying Lie algebra suggests the following definition.

\begin{definition} Let $(C^\infty(M),\{\cdot,\cdot\})$ be a Poisson algebra. A {\it Poisson--Hopf Lie system} is pair consisting of a Poisson--Hopf algebra $C^\infty(\mathfrak{g}_z^*)$ and a $z$-parametrized family of Poisson algebra representations $D_z:C^\infty(\mathfrak{g}_z^*)\rightarrow C^\infty(M)$ with $z\in \mathbb{R}$.
\end{definition}

Next, constants of the motion for ${\bf X}_{z,t}$ can be deduced by applying the coalgebra approach introduced in~\cite{BCHLS13Ham} in the way  briefly described in Section 3. In the deformed case, we consider the Poisson algebra morphisms 
\bea
&& D_z: C^\infty\left( \mathfrak{sl}_z^*(2) \right) \rightarrow C^\infty(M),   \nonumber\\
 && D_z^{(2)} :   C^\infty\left(  \mathfrak{sl}_z^*(2)  \right)\otimes C^\infty\left( \mathfrak{sl}_z^*(2)\right)\rightarrow C^\infty(M)\otimes C^\infty(M),
 \nonumber
\eea 
which by taking into account the coproduct (\ref{gax}) are defined by
\bea 
&& D_z( v_i):= h_{z,i}({\bf x}_1) \equiv   h_{z,i}^{(1)}  , \quad i=1,2,3, \nonumber\\ 
&&  D_z^{(2)} \left( {\Delta}_z(v_1) \right) =h_{z,1}({\bf x}_1)+h_{z,1}({\bf x}_2) \equiv  h_{z,1}^{(2)}   ,  \nonumber \\
&& 
D_z^{(2)} \left( {\Delta}_z(v_k) \right) = h_{z,k}({\bf x}_1)  {\rm e}^{2 z h_{z,1}({\bf x}_2)}  + {\rm e}^{-2 z h_{z,1}({\bf x}_1)} h_{z,k}({\bf x}_2)\equiv  h_{z,k}^{(2)}   ,  \quad k= 2,3,
\nonumber
\eea
where
  ${\bf x}_s$ $(s=1,2)$ are global coordinates in $M$.
We remark that, by construction,  the functions $h_{z,i}^{(2)}$   satisfy the  same Poisson brackets (\ref{ggbb}). 
 Then 
$t$-independent  constants of motion are given by (see (\ref{bc}))
$$
  F_z\equiv F_z^{(1)}:= D_z(C_z),\qquad F_z^{(2)}:=  D_z^{(2)} \left( {\Delta_z}(C_z) \right),
$$
where  $C_z$ is     the deformed Casimir (\ref{gc}). Explicilty, they read
\bea
&&  F_z=  \shc\!\left(2 z h_{z,1}^{(1)}  \right)  h_{z,1}^{(1)}   h_{z,3}^{(1)} - \left( h_{z,2}^{(1)}  \right)^2 = \frac c    4 \,  ,\nonumber\\[2pt]
&& F_z^{(2)}=  \shc\!\left(2 z h_{z,1}^{(2)} \right)  h_{z,1}^{(2)}   h_{z,3}^{(2)} - \left( h_{z,2}^{(2)}  \right)^2  \label{amz} .
\nonumber
\eea


 \section{The three classes of $\mathfrak{sl}(2)$ Lie--Hamilton systems on the plane and their deformation}

 We now apply Theorem 1 to the three  classes of   LH systems in the plane 
with a Vessiot--Guldberg Lie algebra isomorphic to $\mathfrak{sl}(2)$ according to the local classification performed in  \cite{BBHLS}, which was based in the results formerly given in~\cite{GKO92}. Thus the manifold $M=\mathbb R^2$ and the coordinates $\mathbf{x}=(x,y)$. According to  \cite{BBHLS,BHLS}, these three classes are named P$_2$, I$_4$ and I$_5$ and they correspond to a positive, negative and zero value of the Casimir constant $c$, respectively. Recall that these are     non-diffeomorphic, so that  there does not exist  any  local $t$-independent change of variables mapping one into another.


\begin{table}[t]
\caption{The three  classes of LH systems on the plane with underlying Vessiot--Guldberg Lie algebra isomorphic to $\mathfrak{sl}(2)$. For each class, it is displayed, in this order,  a basis of vector fields ${\bf X}_{i}$, Hamiltonian functions $h_i$, symplectic form $\omega$,   the constants of motion $F$ and $F^{(2)}$  as well as the corresponding specific LH systems.}
\begin{tabular}
[l]{l}
\hline
 \noalign{\smallskip}
$\bullet$ Class P$_2$ with $c= 4>0$\\[4pt]
$\displaystyle{\quad  {\bf X}_{1}=
\frac{\partial}{\partial x} \qquad  {\bf X}_{2}=x\frac{\partial}{\partial x}+y\frac{\partial
}{\partial y} \qquad  {\bf X}_{3}=
 (  x^{2}-y^{2} )  \frac{\partial}{\partial x}+2xy\frac{\partial
}{\partial y}}$\\[6pt]
$\displaystyle{\quad  h_{1}= -\frac{1}{y}  \qquad   h_2= -\frac{x}{y}  \qquad  h_3=  -\frac{x^{2}+y^{2}}{y}  \qquad    \omega=\frac{{\rm d}x\wedge
{\rm d}y}{y^{2}} }$  \\[6pt]
$\displaystyle{\quad F=1 \qquad F^{\left(  2\right)  }=\frac{ (  x_{1}-x_{2} )  ^{2}+ (
y_{1}+y_{2} )  ^{2}}{y_{1}y_{2}} }$  \\ 
 \noalign{\smallskip}
-- Complex Riccati equation\\
--  Ermakov system, Milne--Pinney  and Kummer--Schwarz equations with
$c>0$\\
 \noalign{\smallskip}
 \hline
 \noalign{\smallskip}
$\bullet$ Class I$_4$ with $c= -1<0$\\[4pt]
 $\displaystyle{\quad  {\bf X}_{1}= \frac{\partial}{\partial x}+\frac{\partial}{\partial y}\qquad  {\bf X}_{2}= x\frac{\partial
}{\partial x}+y\frac{\partial}{\partial y} \qquad  {\bf X}_{3}=
x^{2}\frac{\partial}{\partial x}+y^{2}\frac{\partial}{\partial y}}$\\[8pt]
$\displaystyle{ \quad  h_{1}=  \frac{1}{x-y}  \qquad  h_2=\frac{x+y}{2 (  x-y )  } \qquad  h_3=\frac{xy}{x-y}\qquad   
\omega=\frac{{\rm d} x\wedge {\rm d} y}{ (  x-y )^{2}  }}$ \\[8pt]   
$\displaystyle{ \quad F=-\frac{1}{4}   \qquad   F^{ (  2 )  }=-\frac{ (  x_{2}-y_{1} )
 (  x_{1}-y_{2} )  }{ (  x_{1}-y_{1} )   (  x_{2} 
-y_{2} )  }  }$\\
     \noalign{\smallskip}
 -- Split-complex Riccati equation\\
--  Ermakov system, Milne--Pinney  and Kummer--Schwarz equations with
$c<0$\\
-- Coupled Riccati equations\\
     \noalign{\smallskip}
 \hline
 \noalign{\smallskip}
$\bullet$ Class I$_5$ with $c=  0$\\[4pt]
 $\displaystyle{\quad  {\bf X}_{1}=  
\frac{\partial}{\partial x}\qquad   {\bf X}_{2}=  x\frac{\partial}{\partial x}+\frac{y}{2} 
\frac{\partial}{\partial y} \qquad   {\bf X}_{3}=   x^{2}\frac{\partial}{\partial x}+xy\frac{\partial}{\partial y}   }$\\[8pt]
$\displaystyle{ \quad  h_{1}=  -\frac{1}{2y^{2}} \qquad  h_{2}=  -\frac{x}{2y^{2}}  \qquad  h_{3}= -\frac{x^{2}}{2y^{2}}   \qquad
\omega = \frac{{\rm d} x\wedge{\rm d}y}{y^{3}}  }$ \\[8pt] 
$\displaystyle{ \quad F=0\qquad F^{\left(  2\right)  }=\frac{ (  x_{1}-x_{2} )  ^{2}}{  
4y_{1}^2y_{2}^2 }  }$\\
     \noalign{\smallskip}
 -- Dual-Study Riccati equation\\
 --  Ermakov system, Milne--Pinney  and Kummer--Schwarz equations with
$c=0$\\
--  Harmonic oscillator\\
-- Planar diffusion Riccati system\\
  \noalign{\smallskip}
\hline
\end{tabular}
\end{table}



\begin{table}[t]
\caption{Poisson--Hopf deformations  of the   three  classes of  $\mathfrak{sl}(2)$-LH systems written in   Table 1.   The symplectic form $\omega$ is the same given in Table 1 and $F\equiv F_z$.}
\begin{tabular}
[l]{l}
\hline
 \noalign{\smallskip}
$\bullet$ Class P$_2$ with $c= 4>0$\\[4pt]
$\displaystyle{   {\bf X}_{z,1}= \frac{\partial}{\partial x} \qquad   {\bf X}_{z,2}=  x\,\mathrm{ch} (  2z/y )
\frac{\partial}{\partial x}+y\,\mathrm{shc}(2z/y)\frac{\partial}{\partial y}     }$  \\[8pt]
$\displaystyle{    {\bf X}_{z,3}=
\left(  x^{2}-\frac{y^{2}}{\mathrm{shc}^{2}(2z/y)}\right)  \mathrm{ch} 
(2z/y)\frac{\partial}{\partial x}+2xy\, \mathrm{shc}(2z/y)\frac{\partial
}{\partial y} }$ \\[8pt]
$\displaystyle{   {\bf h}_{z,1}=  -\frac{1}{y}    \qquad  {\bf h}_{z,2}=  -\frac{x}{y}\,\mathrm{shc} (   {2z}/{y} ) 
\qquad  {\bf h}_{z,3}= -\frac{x^{2} \,\mathrm{shc}^2 (   {2z}/{y} )  +y^{2}}{y\,\mathrm{shc} 
 (   {2z}/{y} ) }}$\\[6pt]
$\displaystyle{   F_z^{(2)}=    \frac{( x_1- x_2)^2 }{ y_1  y_2}\, \shc (2z/  y_1)   \shc (2z/  y_2)   \, \eee^{2z \!/\! y_1}  \eee^{-2z \! /\! y_2} }$\\[8pt]
 $\displaystyle{\qquad\qquad \quad    + \frac{ ( y_1+  y_2)^2}{ y_1  y_2} \,  \frac{\shc^2 (2z/  y_1+ 2z/  y_2) }{  \shc (2z/y_1)   \shc (2z/y_2)  } \,\eee^{2z \! / \! y_1}  \eee^{-2z \!/ \! y_2} }$\\[6pt]
  \noalign{\smallskip}
 \hline
 \noalign{\smallskip}
$\bullet$ Class I$_4$ with $c= -1<0$\\[4pt]
 $\displaystyle{   {\bf X}_{z,1} = \frac{\partial}{\partial x}+\frac{\partial}{\partial y}     }$  \\[8pt]
 $\displaystyle{   
   {\bf X}_{z,2} = 
 \frac{x+y}{2}\,\mathrm{ch}\!\left(\frac{2z}{x-y}\right)\left(  \frac{\partial}{\partial
x}+\frac{\partial}{\partial y}\right)  +\frac{x-y}{2}\,\mathrm{shc}\!\left(\frac
{2z}{x-y}\right)\left(  \frac{\partial}{\partial x}-\frac{\partial}{\partial
y}\right)       }$  \\[8pt]
    $\displaystyle{     {\bf X}_{z,3} = \frac{1}{4} \,\mathrm{ch} \!\left(\frac{2z}{x-y} \right)\left[  (  x+y )
^{2}+ (  x-y )^{2}\mathrm{shc}^{-2}\!\left(\frac{2z}{x-y}\right)\right]  \left(
\frac{\partial}{\partial x}+\frac{\partial}{\partial y}\right)    }$\\[8pt]
 $\displaystyle{\qquad\qquad \quad + \frac{1}{2}\left(  x-y\right)  ^{2}\mathrm{shc}\!\left(\frac{2z}{x-y}\right)\left(
\frac{\partial}{\partial x}-\frac{\partial}{\partial y}\right)  }$\\[8pt]
$\displaystyle{   {\bf h}_{z,1}=   \frac{1}{x-y} \quad  {\bf h}_{z,2}= \frac{(x+y) \,\mathrm{shc}\!\left(  \frac{2z}{x-y}\right)  }{2 (  x-y )  } 
\quad  {\bf h}_{z,3}= \frac{ (  x+y )
^{2}\mathrm{shc}^2\! \left(  \! \frac{2z}{x-y} \right)  - (  x-y )  ^{2} 
}{4 (  x-y ) \, \mathrm{shc}\!\left(\!   \frac{2z}{x-y}\right)  }}$\\[10pt]
$\displaystyle{   F_z^{(2)}= \frac{  ( x_1- x_2+ y_1- y_2)^2   }{4( x_1- y_1)( x_2- y_2)} \, \shc\! \left(\frac{2z} { x_1- y_1}\right)      \shc \!\left(\frac{2z} { x_2- y_2}\right)      \eee^{-\frac{2z}{ x_1- y_1}} \eee^{\frac{2z}{ x_2- y_2}}   }$\\[8pt]
 $\displaystyle{     -  \frac{( x_1+ x_2- y_1- y_2)  \shc\! \left( \!  \frac{2z} { x_1- y_1}+\frac{2z} { x_2- y_2}\right)   }{4( x_1- y_1)( x_2- y_2)}   \left[ \frac{  \eee^{\frac{2z}{ x_2- y_2}}  ( x_1- y_1) }{  \shc \bigl(\frac{2z} { x_1- y_1}\bigr) }   + \frac{    \eee^{-\frac{2z}{ x_1- y_1}} ( x_2- y_2) }{  \shc \bigl(\frac{2z} { x_2- y_2}\bigr) }   \right] }$\\[10pt]
     \noalign{\smallskip}
 \hline
 \noalign{\smallskip}
$\bullet$ Class I$_5$ with $c=  0$\\[4pt]
 $\displaystyle{    {\bf X}_{z,1} =    \frac{\partial}{\partial x}  \qquad  
   {\bf X}_{z,2} =    x\,\mathrm{ch}\!\left(   {z}/{y^{2} 
}\right)  \frac{\partial}{\partial x}+\frac{y}{2}\,\mathrm{shc}\!\left(   
{z}/{y^{2}}\right)  \frac{\partial}{\partial y}  }$   \\[8pt]
    $\displaystyle{     {\bf X}_{z,3} =   x^{2}\,\mathrm{ch}\!\left(   {z}/{y^{2}}\right)  \frac{\partial}{\partial
x}+x y\, \mathrm{shc}\!\left(   {z}/{y^{2}}\right)  \frac{\partial}{\partial
y}      }$\\[6pt]
$\displaystyle{   {\bf h}_{z,1}=  -\frac{1}{2y^{2}}
\qquad  {\bf h}_{z,2}= -\frac{x}{2y^{2}}\,\mathrm{shc}\!\left(   {z}/{y^{2}}\right)  \qquad  {\bf h}_{z,3}=   -\frac
{x^{2}}{2y^{2}} \,\mathrm{shc}\!\left(   {z}/{y^{2}}\right) } $\\[6pt]
$\displaystyle{   F_z^{(2)}=    \frac{( x_1- x_2)^2 }{4  y_1^2  y_2^2} \,\shc \!\left(z/  y_1^2\right)   \shc \!\left(z/  y_2^2\right)  \eee^{z \!/\! y_1^2}  \eee^{-z \! /\! y_2^2} }$ 
     \\[6pt]
   \noalign{\smallskip}
\hline
\end{tabular}
\end{table}

Table 1 summarizes  the three cases, covering vector fields, Hamiltonian functions, symplectic structure and $t$-independent constants of motion. The particular LH systems which are diffeormorphic within each class are also mentioned~\cite{BHLS}. Notice that for all of them  it is satisfy the following commutation relations for the vector fields and Hamiltonian functions (the latter with respect to corresponding $\omega$):
\bea
&&
[{\bf X}_{1},{\bf X}_{2}]={\bf X}_{1},\qquad [{\bf X}_{1},{\bf X}_{3}]=2 {\bf X}_{2},\qquad [{\bf X}_{2},{\bf X}_{3}]={\bf X}_{3},\nonumber\\
&& \{ h_1,h_2\}_\omega = - h_1,\qquad \{ h_1,h_3\}_\omega = -2  h_2,\qquad \{ h_2,h_3\}_\omega = - h_3.
\nonumber
\eea

By applying Theorem 1 with  the results of Table 1 we obtain the corresponding deformations which are displayed in Table 2. 
It is straightforward to verify that the classical 
limit $z\rightarrow 0$ in Table 2 recovers the corresponding starting LH systems and related structures of Table 1, in agreement with   the relations (\ref{zac})  and (\ref{zae}).


\section{A method to construct Lie--Hamilton systems}

Section 5 showed that   deformations of a LH system with a fixed LH algebra $\mathcal{H}_\omega\simeq \mathfrak{g}$ can be obtained through a Poisson algebra $C^\infty(\mathfrak{g}^*)$, a given deformation and a certain Poisson morphism $D:C^\infty(\mathfrak{g}^*)\rightarrow C^\infty(M)$. This section presents a simple method to obtain $D$ from an arbitrary $\mathfrak{g}^*$ onto a symplectic manifold $\mathbb{R}^{2n}$. 

\begin{theorem}\label{MT2} Let $\mathfrak{g}$ be a Lie algebra whose Kostant--Kirillov--Souriau Poisson bracket admits a symplectic foliation in $\mathfrak{g}^*$ with a $2n$-dimensional $\mathcal{S}\subset \mathfrak{g}^*$. Then, there exists a LH algebra on the plane  given by
$$
\Phi:\mathfrak{g}\rightarrow C^\infty\bigl(\mathbb{R}^{2n}\bigr)
$$
relative to the canonical Poisson bracket on the plane.
\end{theorem}
\begin{proof}
The Lie algebra $\mathfrak{g}$ gives rise to a Poisson structure on $\mathfrak{g}^*$ through the Kostant--Kirillov--Souriau bracket $\{\cdot,\cdot\}$. This induces a symplectic foliation on $\mathfrak{g}^*$, whose leaves are symplectic manifolds relative to the restriction of the Poisson bracket. Such leaves are characterized by means of the Casimir functions of the Poisson bracket. By assumption, one of these leaves is $2n$-dimensional. In such a case, the Darboux Theorem warrants that the Poisson bracket on each leave is locally symplectomorphic to the Poisson bracket of the canonical symplectic form on $\mathbb{R}^{2n}\simeq T^*\mathbb{R}^n$. In particular, there exists some Darboux coordinates mapping the Poisson bracket on such a leaf into the canonical symplectic bracket on $T^*\mathbb{R}^n$. The corresponding change of variables into the canonical form in Darboux coordinates can be understood as a local diffeomorphism $h:\mathcal{S}_k\rightarrow \mathbb{R}^{2n}$ mapping the Poisson bracket $\Lambda_k$ on the leaf $\mathcal{S}_k$ into the canonical Poisson bracket on $T^*\mathbb{R}^n$. Hence, $h$ gives rise to a canonical Poisson algebra morphism $\phi_h:C^\infty(\mathcal{S}_k)\rightarrow C^\infty(T^*\mathbb{R}^n)$.  

As usual, a basis $\{v_1,\ldots, v_r\}$ of $\mathfrak{g}$ can be considered as a coordinate system on $\mathfrak{g}^*$. In view of the definition of the Kostant--Kirillov--Souriau bracket, they span an $r$-dimensional Lie algebra. In fact, if $[v_i,v_j]=\sum_{k=1}^rc_{ij}^kv_k$ for certain constants $c_{ij}^k$, then $\{v_i,v_j\}=\sum_{k=1}^rc_{ij}^kv_k$. Since $\mathcal{S}_k$ is a symplectic submanifold, there is a local immersion $\iota:\mathcal{S}_k\hookrightarrow \mathfrak{g}^*$ which is a Poisson manifold morphism. In consequence, 
$$
\{\iota^*v_i,\iota^*v_j\}=\sum_{k=1}^rc_{ij}^k\iota^*v_k.
$$
Hence, the functions $\iota^*v_i$ span a finite-dimensional Lie algebra of functions on $\mathcal{S}$. Since $\mathcal{S}$ is $2n$-dimensional, there exists a local diffeomorphism $\phi:\mathcal{S}\rightarrow \mathbb{R}^{2n}$ and
$$
\Phi:v\in \mathfrak{g}\mapsto \phi\circ\iota^*v\in C^\infty\bigl(\mathbb{R}^{2n} \bigr)
$$
is a Lie algebra morphism.
\end{proof}

Let us apply the above to   explain the existence of three types of LH systems on the plane. We already know that the Lie algebra $\mathfrak{sl}(2)$ gives rise to a Poisson algebra in $C^\infty(\mathfrak{g}^*)$. In the standard basis $v_1,v_2,v_3$ with commutation relations (\ref{brack2}), the Casimir is (\ref{cas2}).
It turns out that the symplectic leaves of this Casimir are of three types: 
\begin{itemize}
 \item A one-sheet hyperboloid when $v_1v_3-v_2^2=k<0$.
 
 \item A conical surface when $v_1v_3-v_2^2=0$.
 
\item A two-sheet hyperboloid when $v_1v_3-v_2^2=k>0$.
\end{itemize}

In each of the three cases we have the Poisson bivector
$$
\Lambda=-v_1\frac{\partial}{\partial v_1}\wedge\frac{\partial}{\partial v_2}-2v_2\frac{\partial}{\partial v_1}\wedge\frac{\partial}{\partial v_3}-v_3\frac{\partial}{\partial v_2}\wedge\frac{\partial}{\partial v_3}.
$$
Then, we have a changes of variables passing from the above form into Darboux coordinates
$$
\bar v_1=v_1,\qquad \bar v_2=-v_2/v_1,\qquad  C= v_1v_3-v_2^2.
$$
Then,
$$
v_1=v_1,\qquad v_2=-\bar v_1 \bar v_2,\qquad  v_3=(C+\bar v^2_1\bar v^2_2)/\bar v_1.
$$ 
On a symplectic leaf, the value of $C$ is constant, say $C=c/4$, and the restrictions of the previous functions to the leaf read
$$
\iota ^*v_1=v_1,\qquad \iota^*v_2=-\bar v_1\bar v_2,\qquad  \iota^*v_3=c/(4\bar v_1)+\bar v_1\bar v^2_2.
$$
This can be viewed as a mapping $\Phi:\mathfrak{sl}(2)\rightarrow C^\infty(\mathbb{R}^2)$ such that
$$
\phi(v_1)=x,\qquad \phi(v_2)=-xy,\qquad \phi(v_3)=c/(4x)+xy^2,
$$
which is obviously a Lie algebra morphism relative to the standard Poisson bracket in the plane. It is simple to proof that when $c$ is positive,  negative  or zero, one obtains three different types of Lie algebras of functions and their associated vector fields span the Lie algebras P$_2$, I$_4$ and I$_5$  as enunciated in \cite{BBHLS}. Observe that since $\phi(v_1)\phi(v_3)-\phi(v_2)^2=c/4$, there exists no change of variables on $\mathbb{R}^2$ mapping one set of variables into another for different values of $c$. Hence, Theorem \ref{MT2} ultimately explains the real origin of all the  $\mathfrak{sl}(2)$-LH systems on the plane.

It is known that $\mathfrak{su}(2)$ admits a unique Casimir, up to  a proportional constant, and the
symplectic leaves induced in $\mathfrak{su}^*(2)$ are spheres. The application of the previous method originates a unique Lie algebra representation, which gives rise to the unique LH system on the plane related to $\mathfrak{so}(3)$. All the remaining LH systems on the plane can be generated in a similar fashion. The deformations of   such Lie algebras will generate all the possible deformations of LH systems on the plane.


 \section{Concluding remarks}

It has been shown that Poisson--Hopf deformations of LH systems based on the simple Lie algebra $\mathfrak{sl}(2)$ can be formulated simultaneously by 
means of a geometrical argument, hence providing a generic description for the deformed Hamiltonian functions and vectors fields, starting from the corresponding
classical counterpart. This allows for a direct determination of the deformed Hamiltonian functions and vector fields, as well as their corresponding Poisson brackets and   commutators, by mere insertion of the data corresponding to the non-deformed LH system.  This procedure has been explicitly illustrated  by obtaining the deformed results 
of Table 2 from the classical ones of Table 1 through the application of Theorem 1.

Moreover we have explained a method to obtain (non-deformed) LH systems related to a LH algebra $\mathcal{H}_\omega$ by using the symplectic foliation in $\mathfrak{g}^*$, where $\mathfrak{g}$ is isomorphic to $\mathcal{H}_\omega$, which has been stated in Theorem 2. This result could further be   applied in order to obtain deformations of LH systems   beyond $\mathfrak{sl}(2)$.   It is also left to accomplish  the deformation of    LH systems     in other spaces of higher dimension.

It seems that the techniques provided here are potentially sufficient to provide a solution to the above mentioned problems. These will be the subject of further work currently in progress.


 \begin{acknowledgement}
A.B.~and F.J.H.~have been partially supported by Ministerio de Econom\'{i}a y Competitividad (MINECO, Spain) under grants MTM2013-43820-P and   MTM2016-79639-P (AEI/FEDER, UE), and by Junta de Castilla y Le\'on (Spain) under grants BU278U14 and VA057U16.  The research of R.C.S.~was partially   supported by grant MTM2016-79422-P  (AEI/FEDER, EU).
E.F.S.~acknowledges a fellowship (grant CT45/15-CT46/15) supported by the  Universidad Complutense de Madrid.   J.~de L.~acknowledges funding from the Polish National Science Centre under grant HARMONIA 2016/22/M/ST1/00542.
\end{acknowledgement}


\end{document}